\documentclass[12pt]{article}%

\nonstopmode

\usepackage{hyperref,url}
\usepackage{tikz-cd}
\usepackage[all]{xy}
\usepackage{tikz}
\usepackage[round]{natbib}
\usepackage{amsmath}
\usepackage{amsbsy}
\usepackage{amscd}
\usepackage{amsopn}
\usepackage{amstext}
\usepackage{amsxtra}
\usepackage{amssymb}
\usepackage{amsthm,amsfonts}
\usepackage{graphicx}
\usepackage{epstopdf}
\usepackage{bbm}
\usepackage{epsfig}
\usepackage[title,titletoc]{appendix}
\usepackage{titling}

\newtheorem{theorem}{Theorem}
\newtheorem{proposition}[theorem]{Proposition}
\newtheorem{lemma}[theorem]{Lemma}

\newtheorem{conjecture}[theorem]{Conjecture}

\newtheorem{definition}{Definition}

\newtheorem{assumption}{Assumption}
\newtheorem{remark}{Remark}

\def\beq{\begin{eqnarray}}
\def\eeq{\end{eqnarray}}
\def\beqq{\begin{eqnarray*}}
\def\eeqq{\end{eqnarray*}}
\def\beeq{\begin{eqnarray*}}
\def\eeeq{\end{eqnarray*}}
\def\be{\begin{equation}}
\def\ee{\end{equation}}
\newcommand{\bea}{\begin{eqnarray}}
\newcommand{\eea}{\end{eqnarray}}
\newcommand{\beaa}{\begin{eqnarray*}}
\newcommand{\eeaa}{\end{eqnarray*}}
\newcommand{\barr}{\begin{array}}
\newcommand{\earr}{\end{array}}

\newcommand{\benum}{\begin{enumerate}}
\newcommand{\eenum}{\end{enumerate}}

\def\BF={{\mathbb{F}}}
\def\BG={{\mathbb{G}}}
\def\BH={{\mathbb{H}}}
\def\bOne={{\bf 1}}

\def\CC{{\cal C}}
\def\CD{{\cal D}}
\def\CE{{\cal E}}

\def\CM{{\cal M}}

\def\CT{{\cal T}}

\def\dee{{\rm d}}

\def\diag{{{\rm diag}}}

\renewcommand{\|}{ { \lvert } }
\def\qed{\hfill$\sqcap\kern-7.0pt\hbox{$\sqcup$}$\\}

\def\BBC{\mathbb{C}}
\def\BBE{\mathbb{E}}
\def\BBP{\mathbb{P}}

\def\BBR{\mathbb{R}}

\def\BBZ{\mathbb{Z}}

\def\Pois{{\rm Pois}}

\def\One{{\bf 1}}

\def\ED{{\rm D}}
\def\TA{{\rm TA}}
\def\TL{{\rm TL}}
\def\IA{{\rm Z}}
\def\ID{{\rm X}}
\def\EE{{\rm E}}

 \def\dee{\mathrm{d}}

\begin{document}

\title{\normalsize \bf SYSTEMIC CASCADES ON INHOMOGENEOUS RANDOM FINANCIAL NETWORKS}
\author{\footnotesize T. R. HURD \\
\footnotesize \it Mathematics \& Statistics, McMaster University, 1280 Main St. West \\
\footnotesize \it Hamilton, Ontario, L8S 4L8, Canada \\
\footnotesize \it hurdt@mcmaster.ca}

%
\date{\footnotesize 
June 30, 2019}
\maketitle

\begin{abstract}

This systemic risk paper introduces inhomogeneous random financial networks (IRFNs). Such models are intended to describe parts, or the entirety, of a highly heterogeneous network of banks and their interconnections, in the global financial system. Both the balance sheets and the stylized crisis behaviour of banks are ingredients of the network model.  A systemic crisis is pictured as triggered by a shock to banks' balance sheets, which then leads to the propagation of damaging shocks and the potential for amplification of the crisis, ending with the system in a cascade equilibrium. Under some conditions the model has ``locally tree-like independence (LTI)'', where a general percolation theoretic argument  leads to an analytic fixed point equation describing the cascade equilibrium when the number of banks $N$ in the system is taken to infinity. This paper focusses on mathematical properties of the framework in the context of Eisenberg-Noe solvency cascades generalized to account for fractional bankruptcy charges.  New results including a definition and proof of the ``LTI property'' of the Eisenberg-Noe solvency cascade mechanism lead to explicit  $N=\infty$ fixed point equations that arise under very general model specifications.  The essential  formulas are shown to be implementable via well-defined approximation schemes, but numerical exploration of some of the wide range of potential applications of the method is left for future work. \\

\noindent{\bf Key words:\ }
Systemic risk, banking network, random financial network, cascade, interbank exposure, funding liquidity,  insolvency, percolation theory, locally tree-like.\\

\noindent{\bf MSC:}
05C80, 91B74, 91G40, 91G50

\end{abstract}
\renewcommand{\thefootnote}{\alph{footnote}}

Ê

Ê
Ê

\section{Introduction}
\label{sec:1}
Systemic risk (SR), the risk of large scale failure of the financial system as defined for example in \cite{Schwarcz08}, has long been understood (see e.g. \cite{Kaufman94}) to involve cascades of contagious shocks of different types, notably funding liquidity shocks such as bank panics and runs, and solvency shocks caused by failed banks. Compared to systems arising in other areas of applied science, the financial system in question is extraordinarily complex in a diversity of aspects. The agents or ``nodes'' will be thought of as banks in this paper, but the picture can easily be extended to allow for different types of  financial institutions, such as firms, funds and households.  The business and trading strategies of a single bank, itself a complex hierarchical entity, are the result of decision making distributed across all the subdivisions of the bank, and are often made under great uncertainty.  The basic ``links'' or ``edges'' of the system, representing exposures, are bi-directional connections between pairs of banks. These change daily and are typically complex arrangements of financial contracts and securities that reference many underlying financial and economic factors. \cite{Haldane09} points out that the legal description of some individual contracts, such CDOs, may run to thousands or millions of pages. The time scales relevant in banking range from microseconds to decades. Banking systems in different countries are strongly linked. On top of this intrinsic complexity, \cite{Corrigan1982} describes how ``banks are special'': They play a critical systemic role at the heart of the much larger macroeconomy. To compound these difficulties, essential systemic data, particularly disaggregated data with counterparty identification, if it exists, is typically only available to regulators, not to observers or the banks themselves. The system is intrinsically opaque. 

A popular approach to understanding financial systemic risk, see \cite{NieYanYorAle07}, \cite{GaiKapa10a}, \cite{AminContMinc16}, \cite{Hurd16a}, \cite{Hurd18}, has been to construct random  financial networks (RFNs), and to explore how different bank behaviour characteristics lead to cascades and amplification of shocks. At the root of such network models is the choice of a random graph distribution to model the ``skeleton graph'' of interconnections between banks. Typically in these models, the edges in these graphs are directed, by convention pointing from debtor to creditor bank.  Some of the popular  choices for the skeleton have been {\em directed configuration graphs} and {\em scale-free graphs}. For example, directed configuration graphs are constructed starting from a ``degree distribution'' that specifies the number of in and out-edges for each node.  So-called inhomogeneous random graphs (IRGs), see \cite{BolSvaRio07} and \cite{vdHofstad16}, are a related class that have been introduced to systemic risk theory more recently in \cite{DeteMeyePana17}  to  capture the diversity of bank sizes, connectivity and types better than the directed configuration graphs normally assumed in this line of SR research.  

The  goal of this paper is to explore how {\em  inhomogeneous random financial networks (IRFNs)}, in other words RFNs whose skeletons are IRGs,  comprising a diverse collection of bank types linked by random exposures, might behave when subjected to the types of crisis triggers and contagion mechanisms often considered in the SR literature. 
A general category of random financial networks will be introduced that consists of a connectivity  ``skeleton'' drawn from the broad class of  inhomogeneous random graphs, on which is defined a random collection of bank balance sheets and interbank exposures. Any large random shock that hits the system is likely to trigger a cascade sequence of secondary shocks converging to a {\em cascade equilibrium} that represents the final outcome of the crisis. This cascade mapping results from the {\em cascade mechanism} that encodes the deterministic behavioural rules banks are assumed to follow during the crisis. 

IRFN models for any value of $N$ can always be explored by pure simulation alone. Alternatively, like configuration graphs, sequences of IRGs parametrized by increasing $N$ can be specified that have an important property called  {\em locally tree-like independence (LTI)}. As described in \cite{Bordenave16} and others, this property implies that the random graph sequence is ``locally weakly convergent'' as $N\to\infty$ to a collection of connected Galton-Watson random trees. The LTI property of IRGs implies for example that for any $k>1$,  the density of cycles of length $k$ in the graph goes to zero as $N$ goes to infinity.   

\cite{AminContMinc16} and \cite{DeteMeyePana17} have proven for certain simpler cascade models on LTI sequences of random skeleton graphs that the large $N$ asymptotics of the cascade equilibrium is determined by a fixed point of a scalar-valued function, whose value can be interpreted as an average default probability. The proof of such results  typically makes use of a combinatorial theorem of \cite{Wormald99}, which is difficult to extend to the more complex situations considered in the present paper. However, the result itself makes intuitive sense in broader generality, because of two manifestations of the LTI property. First,  the LTI property of the skeleton in these models implies it converges (in the ``local weak'' sense) to a collection of random trees. Second, an analogous LTI property can be defined for the cascade mapping itself, which has the meaning that a desired independence structure is exact on all finite random trees. These two properties combined lead to the limiting approximating cascade formulas and therefore it is not surprising that they are exact in the infinite size limit. In general, based on the  symmetries of the Galton-Watson skeleton, these formulas can be shown to boil down to a fixed point equation for a  monotonic function of a particular collection of variables.  It is has  been observed in \cite{MelHacMasMucGle11} that such asymptotic formulas often seem to provide an ``unreasonably effective approximation'' of finite sized systems studied by simulation.  We will follow this common thread in the SR literature,  namely to determine the dependence on the number of banks $N$ in the system, as $N\to\infty$. 

The main contributions of this paper are:
\begin{enumerate}
  \item Introduction of the {\em inhomogeneous random financial network (IRFN) framework}, closely related to the modelling framework of \cite{DeteMeyePana17} , that provides a flexible and scalable architecture for modelling many of the complex network characteristics thought to be relevant to systemic risk.  In particular, we will develop solvency cascade models for networks of banks with arbitrary types.
    \item We present for the first time a cascade analysis for an economically important family of models extending the EN 2001 framework to include {\em partial fractional recovery} of defaulted interbank assets. We also formulate and prove a ``locally tree-like independence property'' for this class of solvency cascade mechanisms. 
    \item A general characterization is provided for the first cascade step in IRFN default models, in the limit $N\to\infty$. 
  \item The large $N$ asymptotics for full solvency cascades that arise in IRFN models is explored. Tractable recursive formulas for {\em cascade equilibria} are formulated and conjectured to hold in the large $N$ limit, based on the LTI property  of both the IRFN itself and the cascade mechanism. \end{enumerate}

The IRFN construction provides two specific benefits compared to the ``configuration graph'' RFN constructions of \cite{GaiKapa10a} and \cite{AminContMinc16}. Firstly, bank type has a direct and basic financial interpretation: logically, a node's degree is dependent on its type. Type is a more intuitive and general notion than node degree, and better suited to SR modelling where edges and degrees are constantly changing while the type of node does not change. Bank types can encode an unlimited range of node characteristics. Secondly, bank type makes better financial sense than node degree as the conditioning random variables determining system dependencies. In our setting, assuming random balance sheets and exposures are independent conditioned on node types is better justified than assuming their  independence conditioned on node degrees. 

The large $N$ arguments developed in this paper can be used to investigate properties that are likely to hold in a wide range of cascade models on large financial networks with the assumed inhomogeneous random graph structure. Such models go far beyond the small class of RFN models for which rigorous asymptotic results have been derived. The heuristic arguments presented here, although conjectural,  will  complement  a well-developed strand of rigorous results in the literature surveyed by \cite{vdHofstad16}, that relate percolation properties on random graphs to properties of branching processes.  Thus this paper presents a road map to developing rigorous percolation methods to prove the conjectures developed in this paper.

Section 2 introduces the concept of {\em inhomogeneous random financial networks} (IRFNs). Also included in the section are some of the probabilistic tools we will use in cascade analysis on an IRFN. Section 3 explores two of the important cascade channels treated in the SR literature, namely solvency cascades and funding liquidity cascades, in terms of {\em cascade mechanisms} that are deterministic rules of behaviour banks are assumed to follow during a crisis. This section then focusses on the cascade mapping that results from a specific solvency cascade mechanism operating within the IRFN  model. Certain consequences of the LTI property are demonstrated supporting the conjecture that some cascade models have a tractable analytic asymptotic $N=\infty$ form. Section 4 provides a brief exploration of some of the issues and the kind of data required to implement the IRFN cascade method for real networks. Finally, a concluding section discusses some of the important questions and next steps to address in order to better understand financial systemic risk. \\

\noindent{\bf Notation:\ } For a positive integer $N$, $[N]$ denotes the set $\{1,2,\dots, N\}$. For a  random variable $X$, its  cumulative distribution function (CDF), probability density function (PDF) and characteristic function (CF) will be denoted $F_X, \rho_X=F_X',$ and $ \hat f_X$ respectively. For any event $A$, $\One(A)$ denotes the indicator random variable, taking values in $\{0,1\}$. 
Any collection of random variables $X=(X_1,X_2,\dots)$ generates a sigma-algebra (or informally ``information set'') denoted by $\sigma(X)$. Landau's ``big O'' notation $f^{(N)}=O(N^\alpha)$ for some $\alpha\in\BBR$ is used for a sequence $f^{(N)}, N=1,2,\dots$ to mean that  $f^{(N)}N^{-\alpha}$ is bounded as $N\to\infty$.

 \section{Defining IRFNs } The financial system at any moment in time will be represented by an object we call an inhomogeneous random financial network, or IRFN. This is the specification of a multidimensional random variable that captures two levels of structure. The primary level of the IRFN, called the {\em skeleton graph}, is the directed random graph with $N$ nodes, which from now on we take to represent ``banks'', and whose directed edges represent the existence of a significant exposure of one bank to another.  The secondary layer specifies the {\em balance sheets} of the banks, including the inter-bank exposures, conditioned on knowledge of the skeleton graph. 
 
 Inhomogeneity in the IRFN model derives from classifying banks by type. The collection of random bank types $\{T_v\}_{v\in[N]}$ will be assumed to completely determine the dependence structure of other random variables. In other words, conditional expectations with respect to the sigma-algebra $\sigma(T):=\sigma(T_v, v\in[N])$ will typically exhibit conditional independence.

\subsection{Skeleton Graph} The skeleton graph is modelled as a directed inhomogeneous random graph (DIRG),  generalizing Erd\"os-Renyi random graphs, in which directed edges are drawn independently between ordered pairs of banks, not with equal likelihood but with likelihood that depends on the bank types. This class has its origins in \cite{ChungLu02} and  \cite{BritDeijMart06} and has been studied in generality in \cite{BolSvaRio07}. For further details about this class, please see the textbooks \cite{vdHofstad16} or \cite{Hurd16a}[Section 3.4]. The DIRG structure arises by the assumption that exposures between counterparties can be treated as independent Bernoulli random variables $I_{vw}$ defined for pairs of banks $(v,w)$, with a probability that depends on their types $T_v,T_w$. 

\begin{assumption}[Skeleton Graph] The primary layer of an IRFN, namely the skeleton graph ${\rm DIRG}(\BBP,\kappa,N)$,  is a directed inhomogeneous random graph with $N$ nodes labelled by $v\in [N]$. It can be defined by two collections of random variables $T_v, v\in[N]$ and $I_{vw}, v,w\in[N]$, with  sigma-algebras $\sigma(T), \sigma(I)$ and $\sigma(T,I)=\sigma(T)\vee \sigma(I)$.  
\begin{enumerate}
\item Nodes: Each node, representing a bank, has type  $T_{v}\in\CT$ drawn independently with probability $\BBP(T)$ from a finite list of {\em types} $\CT:=[M]$ of cardinality $M\ge 2$.  
\item Edges: Directed edges correspond to the non-zero entries of the {\em incidence matrix} $I$. For each pair $v\ne w\in [N]$,  $I_{vw}$ is the indicator for $w$ to be exposed to $v$, which is to say that $v$ has borrowed from $w$. The collection of edge indicators  $I_{vw}$ is an independent family of Bernoulli random variables, conditioned on the type vector $T:=(T_v)_{v\in[N]}$, with probabilities 
  \be\label{PMK} \BBP[I_{vw}=1\mid \sigma(T)]:=\BBP[I_{vw}=1\mid T_v=T,T_w=T']=(N-1)^{-1}\kappa(T,T')\One(v\ne w)\ . \ee

  \end{enumerate}
  \end{assumption}
 Here $\kappa:[M]^2\to[0,\infty)$, the {\em probability mapping kernel}, is assumed to be independent of $N$. It determines the likelihood that two banks $v,w$ of the given types have an exposure edge from $v$ to $w$.  For consistency, we require that $N-1\ge \max_{T,T'}\kappa(T,T')$.

\subsection{Balance Sheets and the Crisis Trigger}\label{triggersection} The additional fundamental assumption of the IRFN modeling framework is that the balance sheets for all banks are derivable from an independent collection of multivariate random variables,  {\em conditioned on the skeleton}. For the types of cascade analysis presented here, balance sheets will be viewed at the coarse-grained resolution as shown in Table \ref{basicbalancesheet2}.

\begin{table}[!htbp]
  \centering 
{\centering\begin{tabular}{|c||c|}
\hline
 Assets  &  Liabilities  \\ \hline\hline inter-bank assets $\bar\IA$
 & inter-bank debt $\bar\ID$ \\ \hline external illiquid assets $\bar{\rm A}$& external debt $\bar\ED$  \\ \hline  external liquid assets $\bar{\rm C}, \bar\Xi$ & equity $\bar\EE,\bar\Delta$  \\
\hline
\end{tabular}}
\caption{A stylized bank balance sheet. }\label{basicbalancesheet2}
\end{table}

Prior to the onset of the crisis, a bank $v$ has a balance sheet that consists of {\em nominal values} \index{nominal value} of assets and liabilities $[\bar \IA,\bar{\rm A}, \bar{\rm C},\bar \ID,\bar{\rm D}, \bar{\rm E}]$ (labelled by barred quantities), which correspond to the aggregated values of the contracts, valued as if all banks are solvent. Nominal values can also be considered {\em book values} \index{book values} or {\em face values}\index{face values}. Assets (loans and securities) and liabilities (debts) are decomposed into internal and external quantities depending on whether the counterparty is a bank or not. The internal assets  $\bar \IA$ and liabilities $\bar \ID$ of the system can be decomposed into the collection of nominal exposures $\bar \Omega_{vw}$.  Banks and institutions that are not part of the system under analysis are deemed to be part of the exterior, and their exposures are included as part of the external debts and assets. Finally, only two categories of external assets are considered. {\em Fixed assets} model the retail loan book and realize only a fraction of their value if liquidated prematurely while {\em liquid assets} include government treasury bills and the like that are assumed to be as liquid as cash.

\begin{definition}\label{balancesheet} The {\em total nominal value of assets} $\overline\TA_v$ of bank $v$ prior to the crisis consists of the {\em nominal internal assets}  ${\bar \IA}_v$, the {\em nominal external illiquid assets} $\bar{\rm A}_v$, and the {\em nominal external liquid assets} $\bar{\rm C}_v$. The {\em total nominal value of liabilities} $\overline\TL_v$ of the bank consists of the {\em nominal internal debt} ${\bar \ID}_v$, the {\em nominal external debt} ${\bar \ED}_v$ and the bank's {\em nominal equity}  ${\bar \EE}_v$. The {\em nominal exposure} of bank $w$ to bank $v$ is denoted by $I_{vw}\bar\Omega_{vw}$. All components of $\bar{\rm B}$ and $\bar \Omega$ are non-negative, and the {\em accounting identities} are satisfied:
\beq
\nonumber &&{\overline \IA}_v=\sum_wI_{wv} \bar \Omega_{wv}, \quad {\bar \ID}_v=\sum_wI_{vw}\bar \Omega_{vw},\quad \sum_v {\bar \IA}_v=\sum_v {\bar \ID}_v,\quad \bar \Omega_{vv}=0\ ,\\
\label{accounting}
&&\overline\TA_v:={\bar \IA}_v+\bar{\rm A}_v+\bar{\rm C}_v={\bar \ID}_v+{\bar \ED}_v+{\bar \EE}_v=: \overline\TL_v\ . \eeq
The independent components of the nominal balance sheet will be denoted by $\bar {\rm B}_v=[\bar{\rm A}_v, \bar{\rm C}_v, {\bar \EE}_v]$.
\end{definition}

A {\em crisis trigger} at a moment in time, which we label by step $n=0$, occurs when a  shock $\delta \rm B=[\delta{\rm A}, \delta{\rm C}, {\delta \EE}]$ to the balance sheets is sufficiently severe to put some banks into a stressed state where not all of their balance sheet entries  ${\rm B}^{(0)}=\bar {\rm B}+\delta {\rm B}$  are positive. For simplicity we assume $\delta\Omega=0,\Omega^{(0)}=\bar\Omega$. To maintain the convention that balance sheet entries are never negative, we introduce {\em buffers} in place of $\bar{\rm C},\bar{\rm E}$. The {\em cash buffer}  $\Xi^{(0)}_v:=\bar{\rm C}_v+\delta{\rm C}_v$ may be negative, in which case the bank $v$ is said to be {\em illiquid}. Similarly, the {\em solvency buffer} $\Delta^{(0)}_v:=\bar{\rm E}_v+\delta{\rm E}_v$ may now  be negative, in which case the bank is said to be {\em insolvent} or, equivalently, {\em bankrupt}. In our general systemic risk modelling paradigm, the {\em  cascade} that follows the crisis trigger will be viewed for $n\ge 0$  as a step-wise dynamics for the collection of balance sheets ${\rm B}^{(n)}_v$ of the entire system as it tries to {\em resolve} these illiquid and insolvent banks. 

Now we make some pragmatic probabilistic assumptions about the initial balance sheet and exposure random variables at $n= 0$, conditioned on the vector of bank types $T=(T_v)_{v\in[N]}$.  Let us denote by $\sigma(T)$ the sigma-algebra generated by $T$.
\begin{assumption}[Balance Sheets and Exposures]
The secondary layer of an IRFN, the collection of initial balance sheets and exposures ${\rm B}^{(0)}_v,\bar\Omega_{vw}$ at step $n=0$,  are continuous random variables that are mutually independent, and independent of $\sigma(I)$,  conditioned on $\sigma(T)$.

\begin{enumerate}
 \item For each bank $v$, the marginal CDF of ${\rm B}^{(0)}_v=[{\rm A}^{(0)}_v,\Xi^{(0)}_v,\Delta^{(0)}_v]$ conditioned on $\sigma(T)$ is an increasing continuous function of $x\in\BBR_+\times\BBR^2$ taking values in $[0,1]$ and depending only on $T_v\in[M]$:  \be
F_{\rm B}(x\mid T_v):=\BBP({\rm B}^{(0)}_v\le x\mid\sigma(T))\label{BalCDF}\ .\ee
 Note that ${\rm A}_v^{(0)}$ is a positive random variable whereas the buffers may be negative. The {\em  initially illiquid banks} are those with $\Xi^{(0)}_v<0$ and  {\em initially insolvent banks} are those with $\Delta^{(0)}_v<0$.
\item For each edge $vw$, the marginal CDF of ${\bar\Omega}_{vw}$ conditioned on $\sigma(T)$ is an increasing function on $\BBR_+=[0,\infty)$ depending only on $T_v,T_w\in[M]$: \be
F_{\Omega}(x\mid T_v, T_w):=\BBP(\bar\Omega_{vw}\le x\mid\sigma(T))\label{OmegaCDF}\ ,\ee
such that
\[F_{\Omega}(0\mid T_v, T_w)=0,\quad \lim_{x\to\infty}F_{\Omega}(x\mid T_v, T_w)=1\ . \]
\end{enumerate}
\end{assumption}

In summary, a finite IRFN representing the system after a crisis trigger amounts to a collection of random variables $(T,I,{\rm B}^{(0)},\bar\Omega)$ satisfying Assumptions 1 and 2.

\subsection{Asymptotic Properties of IRFNs} 

 \subsubsection{Degree Distribution of the Skeleton Graph} 
 
 One is often concerned with the number of counterparties of nodes in directed random graphs, in other words, the in- and out-degree distributions. In DIRG networks, the degree distributions have a natural Poisson mixture structure in the large $N$ limit.   By permutation symmetry, we need only consider bank $1$ with arbitrary type $T_1=T$, whose  in/out degree  is defined as the pair $(\dee_1^-,\dee_1^+)=\sum_{w=2}^N (I_{w1},I_{1w})$, a sum of conditionally IID bivariate random variables.  Each term has the identical bivariate conditional characteristic function
 \beeq \BBE^{(N)}[e^{ik_1I_{w1}}e^{ik_2I_{1w}}\mid T_1=T]&=&\sum_{T'\in[M]}\BBP(T')\left(1+(N-1)^{-1}\kappa(T,T')(e^{ik_1}-1)\right)\\
 &&\times\ \left(1+(N-1)^{-1}\kappa(T',T)(e^{ik_2}-1)\right)\ .
 \eeeq
 
The conditional CF of $(\dee_1^-,\dee_1^+)$ is the $N-1$ power of this function, and dropping higher order terms in $N^{-1}$ this can be written
\beq \BBE^{(N)}[e^{ik_1\dee^-_{1}+ik_2\dee^+_{1}}\mid T]&=&\label{degreeCF}\\
&&\hspace{-2in}\left[1+\frac1{N-1}\sum_{T'}\BBP(T')\left( \kappa(T,T')(e^{ik_1}-1)+\kappa(T',T)(e^{ik_2}-1)\right)+O(N^{-1})\right]^{N-1}\ ,\nonumber
\eeq
 which displays  simple asymptotic structure as $N\to\infty$.   
 \begin{proposition}\label{Prop1} The characteristic function of the joint in/out degree $(\dee_v^-,\dee_v^+)$ of a bank $v$, conditioned on its bank-type $T\in[M]$, is $2\pi$-biperiodic   on $\BBR^2$ and has the $N\to\infty$ limiting behaviour:    
\beq \label{prop1}\hat f^{(N)}(k_1,k_2\mid T)&=&\hat f(k_1,k_2\mid T)\   \left(1+O(N^{-1}) \right)\ ,\\
\hat f(k_1,k_2\mid T)&:=& \exp\left[\lambda^-(T)(e^{ik_1}-1)
  +\lambda^+(T)(e^{ik_2}-1)\right]\ ,
 \nonumber \eeq
  where $\lambda^+(T)=\sum_{T'}\BBP(T')\kappa(T',T), \lambda^-(T)=\sum_{T'}\BBP(T')\kappa(T,T')$. Here,  convergence of the logarithm of \eqref{prop1} is in $L^2([0,2\pi]\times[0,2\pi])$.
  \end{proposition}

This type of limit can be handled by the following technical lemma, proved in the Appendix.
 
 \begin{lemma}\label{limitlemma} Let $\bar y>0$ and $I$ be any hyperinterval in $\BBR^d$. Suppose $g(x,y):I\times[0,\bar y]\to\BBC$ is a bivariate function such that  $g(\cdot,y), \partial_y g(\cdot,y), \partial^2_y g(\cdot,y)$ are pointwise bounded and in $L^2(I)$ for each value $y\in[0,\bar y]$. Then 
 \[ \lim_{y\to 0} \||\frac1{y}\log(1+yg(x,y))]- g(x,0)\||_{L^2}=O(y)\ .\]
\end{lemma}

\begin{proof}(Proposition 2) Apply Lemma \ref{limitlemma} to $\log\left(\BBE^{(N)}[e^{ik_1\dee^-_{1}+ik_2\dee^+_{1}}\mid T]\right)$ for each $T$ with $N-1=y^{-1}$ and 
\beqq g(k_1,k_2,y)&=&\sum_{T'\in[M]}\BBP(T')\Bigl[\kappa(T,T')(e^{ik_1}-1)+\kappa(T',T)(e^{ik_2}-1)\\
&&+y\kappa(T,T')(e^{ik_1}-1)\kappa(T',T)(e^{ik_2}-1) \Bigr]\ .
\eeqq
\end{proof}

Thus, for different values of $T$, the conditional joint in/out degree distribution is always asymptotic to a bivariate Poisson distribution. Now, recall that a {\em finite mixture} of a collection of probability distribution functions is the probability formed by a convex combination. We can then see that under our simplification of a finite type space $[M]$, the asymptotic {\em unconditional} in/out degree distribution of any bank is a finite mixture:
\[  \hat f^{(N)}(k_1,k_2)= \sum_{T}\BBP(T) \hat f^{(N)}(k_1,k_2\mid T)\ ,\]
 where each component has a bivariate Poisson distribution with Poisson parameters
  \[  \Bigl(\sum_{T'}\BBP(T')\kappa(T',T), \sum_{T'}\BBP(T')\kappa(T,T_w)\Bigr)\ .\]
 The mixing variable is the bank-type $T$, with the mixing weights $\BBP(T)$.

Proposition \ref{Prop1} is a manifestation of the locally tree-like property of IRFNs. Consider the limiting distribution of the interbank debt $\bar\ID_1=\sum_{w\ne 1}\bar\Omega_{w1}$ of a typical bank $v=1$.  
\begin{proposition} \begin{enumerate}
  \item The characteristic function of the interbank debt $\bar\ID_1$ of bank $1$, conditioned on its bank-type $T\in[M]$, has the $N\to\infty$ limiting behaviour:    
 \beq\label{prop3} \hat f^{(N)}_\ID(k\mid T)&:=&\BBE^{(N)}\left[\prod_{w\ne 1}e^{ikI_{w1}\bar\Omega_{w1}}\mid T\right]=\hat f_X(k\mid T)(1+O(N^{-1})) \\
\hat f_X(k\mid T)&=&
 \exp\left[\sum_{T'}\BBP(T')\kappa(T',T)(\hat f_{\bar\Omega}(k\mid T',T)-1)\right]
\label{hatfX}\eeq
   where convergence of the logarithm of \eqref{prop3} is in $L_2[0,\infty)$.
 \item Any finite collection of interbank debt random variables $\{\bar\ID_v, v\in [1,2,\dots, P]\}$ is independent in the $N\to\infty$ limit.\end{enumerate}
\end{proposition}

\begin{proof} By the conditional independence of the factors, we have an exact formula valid for finite $N$:
\beq
\hat f^{(N)}_\ID(k\mid T)&=&\prod_{w\ne 1}\BBE^{(N)}[1+ I_{w1}(e^{ik\bar\Omega_{w1}}-1)\mid T]\label{fXN}\\
&=&\left(1+\sum_{T'}\BBP(T')\frac{\kappa(T',T)}{N-1}\bigl(\hat f_{\Omega}(k\mid T',T)-1\bigr)\right)^{N-1}\nonumber
\eeq
Now, by applying Lemma \ref{limitlemma} to $\log  \hat f^{(N)}_\ID$ with $N-1=y^{-1}$ and\\ $g(k,y)=\sum_{T'}\BBP(T')\kappa(T',T)(\hat f_{\Omega}(k\mid T',T)-1)$, the limit  in $L_2[0,\infty)$ is 
\[\log \hat f^{(N)}_\ID(k\mid T)=\log\hat f_X(k\mid T)+O(N^{-1})\]
where $\hat f_X(k\mid T)$ is as stated above.
\end{proof}

\begin{remark} Comparison of equation \eqref{hatfX} to the L\'evy-Khintchin formula shows that $\bar X$ is a positive compound Poisson random variable with a continuous jump measure $\dee \mu_X(\cdot\mid T)$ on $\BBR^+$:
\beq \hat f_X(k\mid T)&=&\exp\left[\int^\infty_0 [e^{iku}-1]\mu_X(u\mid T)\dee u\right]\ ,\\
\mu_X(u\mid T)&=& \sum_{T'}\BBP(T')\ \kappa(T',T)\ \rho_\Omega(u\mid T',T)\ .\eeq
It follows that the unconditional distribution of $\bar\ID_v$ is a {\em mixture} over $T_v$ of compound Poisson random variables, with mixing distribution $\BBP(T_v)$, including a positive probability\\ $\sum_T\BBP(T)e^{-\int^\infty_0\mu_X(u\mid T)\dee u}$ for $X=0$. 
\end{remark}

For part (2), note that the same proof implies that the joint conditional CF of $\bar\ID_1,\bar\ID_2$  two banks will be given by 
\[\BBE^{(N)}(e^{ik_1\bar\ID_1}e^{ik_2\bar\ID_2}\mid T_1,T_2)=\hat f^{(N)}_\ID(k_1\mid T_1)\hat f^{(N)}_\ID(k_2\mid T_2) \left(1+O(N^{-1})\right)\ .\]
Similarly, for the joint conditional CF for any finite collection of banks. 

\subsection{The Galton-Watson Process}\label{GW} 
The result in Proposition 2 on the large $N$ asymptotic degree distribution reflects the general principle discussed in \cite{Bordenave16} that a sequence of locally tree-like networks such as an IRG is always ``locally weakly convergent'' to a collection of connected Galton-Watson (GW) random trees. This has a well-defined meaning that the collection of nodes that can be reached from a given node of type $T$ by following directed edges has the approximate structure of a branching process. We can interpret Proposition 2 as implying that the number of nodes of type $T'$ that can be reached along single directed edges rooted at any node $v$ with type $T$ is a Poisson random variable $X_{T',T}$ with mean parameter  $\BBP(T')\kappa(T,T')$.  Each subsequent one-step extension has the same underlying distribution, defining the branching process.  These facts identify the offspring distribution of any node, conditioned on its type. 
 
It is not hard to deduce that $Z_{n,T',T}$, the number of $n$ step directed paths rooted at the node $v=1$ with type $T$ and terminating in a  node of type $T'$, is a random variable that will follow the recursion formula
 \be\label{GWrecursion} Z_{n,T',T}=\sum_{T''\in[M]} \sum_{i=1}^{Z_{n-1,T'',T}} X_{T',T'',i}\ , n>1\ee
 with $Z_{1,T',T}\sim X_{T',T,1}$. Here $\{Z_{n-1,T'',T}, X_{T',T'',i}\}_{i\in\BBZ_+}$ is a mutually independent collection of random variables and each $X_{T'',T,i}$ is identically $\Pois(\BBP(T'')\kappa(T,T''))$. 
 
We now provide a multi-type extension of the discussion of branching processes found in \cite{Hurd16a}[Section 4.1]. Let ${\bf G}=(G_1,\dots, G_M):[0,1]^M\to[0,1]^M$ denote the following probability generating function for the identically distributed multi-variate random variables  $X_{T',T,i}$: for  ${\bf a}=(a_1,\dots, a_M)$ and $T\in[M]$  the $T$-component of $\bf G$ is defined to be
\be G_T({\bf a})=\BBE \prod_{T'=1}^M (a_{T'})^{X_{T',T}}=\exp[\sum_{T'} \BBP(T')\kappa(T,T')(a_{T'}-1)] .\ee
From  the GW recursion \eqref{GWrecursion}, one can verify that the probability generating functions ${\bf H}_{n}=(H_{n,1},\dots, H_{n,M}):[0,1]^M\to[0,1]^M$ for the multi-variate random variables  $Z_{n,T',T}$ defined by 
\be H_{n,T}({\bf a})=\BBE \prod_{T'=1}^M (a_{T'})^{Z_{n,T',T}}\ee
are given by the composition 
\be {\bf H}_{n}= {\bf H}_{n-1}\circ  {\bf G}=\underbrace{{\bf G}\circ {\bf G}\cdots \circ {\bf G}}_{n\ {\rm factors}}\ . \ee

The extinction probabilities for this GW process forms a vector $\xi=(\xi_T)_{T\in[M]}$ where $\xi_T:=\BBP[\ \exists\ n: Z_{n,T',T}=0\ \forall\ T']$. For each $n$, define $\xi_{n,T}=\BBP[Z_{n,T',T}=0\ \forall\ T']=H_{n,T}({\bf 0})$. Since $ Z_{n-1,T',T}=0 \ \forall \ T',T$ implies $Z_{n,T',T}=0\ \forall\ T',T$,  the sequence $\xi_{n}=(\xi_{n,1},\xi_{n,M})$ is increasing, bounded  and therefore converges to some value $\xi\in [0,1]^M$. Since $\xi_n={\bf H}_n({\bf 0})$, 
\[\xi_n={\bf G}({\bf H}_{n-1}({\bf 0}))={\bf G}(\xi_{n-1})\ .
\]
Note also that ${\bf G}$ has ${\bf G}(\One)=\One$ and is continuous and increasing  on $[0,1]^M$. Therefore, by continuity,
\[\xi=\lim_{n\to\infty}\ \xi_n = \lim_{n\to\infty}{\bf G}( \xi_{n-1})={\bf G}\left(\lim_{n\to\infty} \xi_{n-1}\right)={\bf G}(\xi)\ ,
\]
so $\xi\in[0,1]^M$ is a fixed point, which we can also see is the least fixed point, of ${\bf G}$. Since ${\bf G}$ is strictly convex everywhere, it can have at most two fixed points on the lattice $[0,1]^M$.  
From this discussion, one can deduce that a node with type $T$ will have an infinite number of nodes in its forward cluster with probability $1-\xi_T$, which will be non-zero for some $T$ if the gradient $\nabla {\bf G}$ at $\xi=\One$ has its maximal eigenvalue greater than $1$.  

This recaps the main result of percolation theory that the existence or not of an infinite connected cluster is directly related to whether or not the extinction probability vector is $\xi<\One$ or not. Since $\xi$ is the least fixed point of $\bf G$, this amounts to the existence or not of a non-trivial fixed point of the analytic function ${\bf G}:[0,1]^M\to[0,1]^M$. Our object now is to define how financial crises can be modelled as cascades on random financial networks. Percolation theory, as an abstract exploration of network connectivity, is a guide to understanding the susceptibility of such financial networks to cascades.

\section{Default Cascades on IRFNs}

The IRFN framework specifies the distributions of the random variables $T,I,{\bar {\rm B}},\bar\Omega$ just introduced. It provides a compact stochastic representation of the state of a given real world network of $N$ banks at a moment in time prior to a crisis. With the same distributional data, we can consider this as an element of a sequence of networks by varying $N$ and taking $N\to\infty$.   We now want to consider  how such networks will respond when a {\em trigger event} at time $t=0$ moves the pre-trigger balance sheets $\bar {\rm B}=[\bar{\rm A},\bar \Xi,\bar \Delta]$ to the post-trigger balance sheets ${\rm B}^{(0)}=\bar {\rm B}+\delta {\rm B}$ (recall we assume $\Omega^{(0)}=\bar\Omega$). 

Cascade mechanisms (CMs) are stylized behaviours that banks are assumed to follow when they become aware that a crisis has been triggered. These behaviours are highly non-linear, to reflect that during a crisis banks will take emergency or remedial actions, and in the worst case of bankruptcy be taken over by a {\em system regulator}. ``Business as usual'', in which banks react smoothly to small changes as they pursue profits, is  not applicable during the crisis. Instead we assume healthy banks that are solvent and liquid adopt a ``do nothing/wait and see'' crisis management strategy, while weak banks' behaviour may be forced or severely constrained by the regulatory structure. From a systemic perspective, cascades can arise when weak banks' behaviour have negative impact on other banks.  

\cite{Hurd18} provides an overview of some of the important cascade channels that model the forced behaviour of banks when their buffers fall below certain thresholds. For example,  {\em funding liquidity cascades} arise when banks experience withdrawals by depositors or wholesale lenders. After $n$ steps of the cascade their impacted cash buffers will be $\Xi^{(n)}_v=\Xi^{(0)}_v-\sum_w \tilde S^{(n-1)}_{wv}$ where $\tilde S^{(n-1)}_{wv}$ denotes the {\em  liquidity shock} transmitted from bank $w$ hitting bank $v$. This section focuses instead on {\em solvency cascades}, which turn out to have the same mathematical structure as funding liquidity cascades. In this channel the most relevant buffer variable is the impacted solvency buffer after $n$ cascade steps $\Delta^{(n)}_v=\Delta^{(0)}_v-\sum_w S^{(n-1)}_{wv}$, where $S^{(n-1)}_{wv}$ denotes the {\em solvency shock} from $w$ to $v$. 

\subsection{Default Cascade Mechanisms}
We now consider a class of  models generalizing the clearing model for defaulted debt of \cite{EiseNoe01}. The original EN model assumes that no bankruptcy charges are paid when a firm defaults, ruling out a dangerous contagion amplification mechanism. More realistically, bankruptcy charges and frictions will likely amount to a substantial effective cut of the firm's value at its default. \citet{RogeVera13} extend the EN model in this direction by assuming that bankruptcy costs given default are linear in the endowment and the recovery value of interbank assets. In their model, the recovery value is discontinuous in buffer variables at the solvency threshold, creating an effectively infinite shock amplification effect at this ``hard threshold''. In contrast, we make a  ``soft threshold'' assumption  where the recovery fraction on interbank debt is a continuous piecewise linear function of the level of insolvency. 

Partial recovery of the notional value of the defaulted bank's assets will therefore be assumed to be distributed amongst creditors according to their seniority.  Banks are assumed to have balance sheets as in Table \ref{basicbalancesheet2}, and to be insolvent (bankrupt) if and only if $\Delta<0$.    

\begin{assumption}[Fractional Recovery] For each bank,  \begin{enumerate}
  \item External debt $\ED$ is senior to interbank debt $\ID$ and all interbank debt is of equal seniority;
  \item Bankruptcy charges are in proportion to the negative part of the impacted solvency buffer. 
\end{enumerate} 
\end{assumption}
Thus there is a fixed parameter $\lambda\in(0,1]$ assumed to be the same for all banks, such that at step $n$ of the cascade  
\be\label{bankruptcycosts}\mbox{bankruptcy costs} = (1/\lambda-1)\max(-\Delta^{(n)},0)\ . 
\ee 
This assumption implies that the amount available to repay all debtors of a defaulted firm is $\overline{\rm TA}-(1/\lambda-1)\max(-\Delta^{(n)},0)$. It means that as soon as $\Delta^{(n)}_v\le -\lambda  \bar\ID_v$, the recovery fraction paid on defaulted interbank debt will be zero. 

In general, the {\em loss fraction on interbank debt} of each bank at step $n$ can be identified as the 
 {\em insolvency level} random variable defined by 
\be\label{solvency1}\CD^{(n)}_v=g_\lambda\left(\frac{\Delta^{(n)}_v}{\bar \ID_v}\right),\quad g_\lambda(x):=\min\left(1,\max(-x/\lambda,0)\right)\ .\ee 
The insolvency level of bank $w$ at step $n$ now influences the solvency shock transmitted to another bank $v$:
\be\label{stepn:1}
S^{(n)}_{wv}:=I_{wv}\bar \Omega_{wv}\CD^{(n)}_w\ ,\ee
 the aggregated solvency shock transmitted to $v$:
\be \label{stepn:2}S^{(n)}_{v}:= \sum_{w\ne v} S^{(n)}_{wv}\ ,\ee
and finally,   the solvency buffer of $v$ at the end of step $n$:
\be\label{impactedsolvency2}
\Delta^{(n+1)}_v=\Delta^{(0)}_v-\sum_w S^{(n)}_{wv}\ .
\ee 
Putting (\ref{solvency1},\ \ref{stepn:1},\ \ref{stepn:2},\ \ref{impactedsolvency2}) together gives the complete solvency cascade mapping at step $n\ge 0$.

\subsection{The First Cascade Step}

Consider \eqref{stepn:1} for $n=0$ defining the single shock $S^{(0)}_{21}=I_{21}\bar \Omega_{21}\CD^{(0)}_2$ transmitted from 2 to 1 for two typical banks $1,2$. Then $S^{(0)}_{21}=G_\lambda(X,Y,Z)$ where the {\em shock transmission function} \be\label{Glambda} G_\lambda(x,y,z)=zg_\lambda(y/(x+z)), \ g_\lambda(u)=\min(1,\max(-u/\lambda,0))\ .\ee
depends on the independent random variables $X= X_{2\backslash 1}:=\sum_{w\ne 1, 2} I_{2w}\bar\Omega_{2w}, Y:=\Delta^{(0)}_2, Z:= I_{21}\bar\Omega_{21}$. The next proposition shows that the characteristic function of $S^{(0)}_{21}$ for finite $N$ can be expressed  in terms of 
 \be\label{RNdef} R^{(N)}(k, k'\mid T,T'):=\frac1{2\pi}\int^0_{-\infty} e^{ik'y} \BBE^{(N)}[e^{ikG_\lambda(X,Y,\bar\Omega_{21})}-1\mid Y=y, T_1=T,T_2=T']\ \dee y\ ,
\ee
 the Fourier transform of the conditional characteristic function of a particular random variable related to $S^{(0)}_{21}$.

\begin{proposition} The  characteristic function of the  {\em solvency shock} $S^{(0)}_{21}$ transmitted from bank $2$ to bank $1$ in step $0$, conditioned on the types $T_1=T, T_2=T'$, is given for finite $N$ by
 \be \BBE^{(N)}[e^{ikS^{(0)}_{21}}\mid T,T']=1+\frac{\kappa(T',T)}{N-1} \int^\infty_{-\infty}\ \hat f^{(0)}_\Delta(k'\mid T') \ R^{(N)}(k, k'\mid T,T')\ \dee k'\ .
 \label{shock1}\ee

\end{proposition}

\begin{proof} The proof works for any bounded shock transmission function $G$ such that $G(x,y,0)=0$ and $G(x,y,z)\One(y\ge 0)=0$.  Since $e^{ikG(X,Y,I_{21} \bar\Omega_{21})}=1+I_{21}(e^{ikG(X,Y, \bar\Omega_{21})}-1)$ and  $(e^{ikG(X,Y, \bar\Omega_{21})}-1)\One(Y\ge 0)=0$, 
 \beqq \BBE^{(N)}[e^{ikS^{(0)}_{21}}\mid T,T']&=&1+\BBE^{(N)}[I_{21}(e^{ikG(X,Y,\bar\Omega_{21})}-1)\ \One(Y<0)\mid T,T']\\
&&\hspace{-1.2in}=\ 1+ \frac{\kappa(T',T)}{N-1}\int^0_{-\infty} \rho^{(0)}_\Delta(y|T')\BBE^{(N)}[e^{ikG(X,Y,\bar\Omega_{21})}-1\mid Y=y, T,T']\ \dee y
 \eeqq
 which by \eqref{RNdef} and  the Parseval-Plancherel identity in Fourier analysis yields  the required result \eqref{shock1}.  
 \end{proof}
 
While the proposition provides an abstract characterization of the result  for  general shock transmission functions $G(X,Y,Z)$, we will also need explicit integral formulas for the specific function $G_\lambda$ given by \eqref{Glambda}. Fix $a=-y/\lambda$ and define 
 \[ R^{(N)}(k,a):=\BBE^{(N)}[e^{ikG_\lambda(X,Y,\bar\Omega_{21})}-1\mid Y=-a\lambda, T,T']\ ,\]
 and note that $R^{(N)}(k,a)=0$ for $a\le 0$. For $a>0$,  taking into account that $\BBP^{(N)}(X=0)\ne 0$, we can write
 \beqq R^{(N)}(k,a)&=& \BBP^{(N)}(X=0)\int_{\BBR_+}\rho_\Omega(z) (e^{ikG_\lambda(0,-a\lambda,z)}-1)\ \dee z\\
 &&+\iint_{\BBR^2_+}\rho^{(N)}_X(x)\rho_\Omega(z) (e^{ikG_\lambda(x,-a\lambda,z)}-1)\ \dee x\ \dee z\eeqq
where $\int_{\BBR_+}\rho^{(N)}_X(x)\dee x=1-\BBP^{(N)}(X=0)$.  In the double integral we change integration variables to $(x,u)\in\BBR_+\times[0,a]$. This has the inverse transformation \[(x,z)=\Bigl(x, u\One(x\le a-u)+\frac{ux}{a-u}\One(x>a-u)\Bigr)\ . \]
 After some manipulation, this gives the  formula
  \beq R^{(N)}(k,a)&=&\BBP^{(N)}(X=0)\BBP(\bar\Omega_{21}>a)(e^{ika}-1)\nonumber\\
  &+&\ \int^a_0\BBP^{(N)}(X\in[0, a-u))\rho_\Omega(u) (e^{iku}-1)\ \dee u\nonumber\\
  &+&\ \int^a_0\left[ \int^\infty_{a-u} \rho_\Omega\Bigl(\frac{ux}{a-u}\Bigr)\frac{ax}{(a-u)^2}\rho^{(N)}_X(x)\  \dee x\right]\ (e^{iku}-1)\ \dee u\ . \label{Rformula}\eeq
Finally, we note that the distribution of $X=X_{2\backslash 1}$ for any $N$ can be computed using Proposition 3 with a replacement of $N$ by $N-1$, which leads to an explicit multi-dimensional integral for $R^{(N)}(k, k'\mid T,T')$.

We next consider the asymptotic distribution of the total solvency shock $S^{(0)}_1:=\sum_{w\ne 1} S^{(0)}_{w1}$ transmitted to bank $1$ in step $0$. By a slight generalization of Lemma \ref{limitlemma}, one can argue that any {\em finite collection} of shocks $\{S^{(0)}_{w1}\}_{w\ne 1}$ are identical, and asymptotically independent, conditioned on the type $T_1=T$. However, this fact cannot prove the following plausible statement:
\beqq \BBE^{(N)}[e^{ikS^{(0)}_{1}}\mid T]&=& \BBE^{(N)}\left[\prod_{w\ne 1}e^{ikS^{(0)}_{w1}}\mid T\right]\\&&\hspace{-1.7in}\sim\prod_{w\ne 1}\BBE^{(N)}[e^{ikS^{(0)}_{w1}}\mid T](1+O(N^{-1}))=\left(\sum_{T'}\BBP(T')\BBE^{(N)}[e^{ikS^{(0)}_{21}}\mid T,T']\right)^{N-1}(1+O(N^{-1}))
\eeqq
where $\sim$ represents the unproven step. Accepting this unproven step as true, and following the argument proving Proposition \ref{Prop1}  leads to the following conjecture:

\begin{conjecture}\label{Conj}
 The characteristic function of the {\em  total solvency shock} $S^{(0)}_{1}=\sum_{w\ne 1} S^{(0)}_{w1}$ transmitted to bank $1$ in step $0$, conditioned on the type $T_1=T$, has   the $N\to\infty$ limiting behaviour:    
  \beq \BBE^{(N)}[e^{ikS^{(0)}_{1}}\mid T]&=&\label{shock2}\hat f^{(0)}_{S}(k\mid T)(1+O(N^{-1}))\ ,\\
 \hat f^{(0)}_{S}(k\mid T) &:=&\exp \left(\sum_{T'}\BBP(T')\kappa(T',T) \int^\infty_{-\infty} \ \hat f^{(0)}_\Delta(k'\mid T')\ R(k, k'\mid T,T')\ \dee k'\right)\nonumber
 \eeq
  where the limit of the logarithm  is in $L^2[0,\infty)$. Here
  \[ R(k, k'\mid T,T'):=\frac1{2\pi}\int^0_{-\infty} e^{ik'y} R(k,-y/\lambda)\ \dee y
\]
where $R(k,a)$ is given by \eqref{Rformula} with $N=\infty$ and the conditions $T,T'$.
\end{conjecture}  

An interpretation of this conjecture based on the formula \eqref{Rformula} is that the solvency shock hitting bank $v=1$ in step $0$ is a non-negative L\'evy distributed random variable. Moreover, the jump measure is a specific {\em non-linear} convolution of the three component probability density functions. 

On the right side of  equation \eqref{impactedsolvency2} for the impacted default buffer  $\Delta^{(1)}_1=\Delta^{(0)}_1-S^{(0)}_1$ at the end of step $0$ we see directly that  $S^{(0)}_1$ and $\Delta^{(0)}_1$ share no common balance sheet random variables, and are therefore independent conditionally on the type $T$ of bank $1$. From the multiplicative property of characteristic functions of sums of independent random variables, the impacted default buffer  $\Delta^{(1)}_1$ has the product conditional characteristic function 
\be\label{newdelta} \hat f^{(1)}_{\Delta}(k\mid T) =\hat f^{(0)}_{\Delta}(k\mid T) \hat f^{(0)}_{S}(-k\mid T)\ .
\ee

In summary,  step $0$ of the solvency cascade mapping has been broken down into three substeps that capture the probabilistic implications of equations \eqref{stepn:1}-\eqref{impactedsolvency2}. 
Each of these substeps depends on the initial conditional distributional data for the collection $\{T_v, I_{vw}, \bar\Omega_{vw},\Delta^{(0)}_v\}$, combined with a conditional independence assumption. The result of the mapping is full conditional univariate distributional data for the collection $\{\Delta^{(1)}_v\}$.  

\subsection{LTI Cascade Mechanisms}
It turns out that Conjecture \ref{Conj} is  understandable from a different perspective if we consider the solvency cascade mapping on an IRFN, under the condition that the skeleton graph is a non-random, connected, directed tree. This alternative line of thinking is motivated because the skeleton of an IRFN is an IRG, which we have observed will always converge ``locally weakly'' to a random ensemble of connected components which are {\em trees}. We will now prove that the solvency cascade mapping on a skeleton which is  a non-random, connected, directed tree has a nice property we call {\em locally tree-like independent}.  

Let $([N],\CE)$ denote the nodes and edges of such a tree skeleton, with node types labelled. Being a connected, directed tree, there is a partial ordering $\ge, >$ generated by the relationships $w>(wv)>v$ whenever $(wv)\in\CE$.  Every element of $([N],\CE)$ is connected to a fixed node $w$ by a unique path, whose final edge is either into or out of $w$. For any collection $A$ of nodes and edges, we denote by $\sigma(A)$ the sigma-algebra generated by the collection of random variables $\{ \Delta^{(0)}_u, \bar\Omega_{wv}\}_{u,(wv)\in A}$. When $([N],\CE)$ arises from an IRFN on a tree, this is always a mutually independent collection. If $A$,$B$ are disjoint subsets, then $\sigma(A)$ and $\sigma(B)$ are always independent. Now, for each $u\in[N]$  and $(wv)\in\CE$ we define some natural collections of random variables and their sigma-algebras.  
\begin{enumerate}
  \item $\CM^-_u$: the subset of $([N],\CE)$ whose elements are each connected to $u$ by a  path whose final edge is directed into $u$.
   \item $\CM^+_u$: the subset of $([N],\CE)$ whose elements are each  connected to $u$ by a  path whose final edge is directed out of $u$. 
  \item $\CM^-_{v\backslash w}$: the subset of $([N],\CE)$  whose elements are each connected to $v$ by a  path whose final edge is directed into $v$, but is not the edge $(wv)$.
   \item $\CM^+_{w\backslash v}$: the subset of $([N],\CE)$ whose elements are each connected to $w$ by a  path whose final edge is directed out of  $w$, but is not the edge $(wv)$.
\end{enumerate}
Note that the following are disjoint unions for all $u\in[N], (wv)\in\CE$:
\be\label{disjoint}   [N]\cup\CE=\CM^-_u\cup\CM^+_u\cup\{u\}\ ;\quad 
\CM^-_v=\CM^-_{v\backslash w}\cup\{(wv)\}\cup\{w\}\cup\CM^+_{w\backslash v}\cup \CM^-_{w}\ .\ee

\begin{definition}
A solvency cascade mechanism has the  {\em  locally tree-like independent property} if, conditioned on the skeleton being a non-random, connected, directed tree,
$\Delta^{(n)}_{v}$ is $\sigma(\CM^-_v\cup\{v\})$-measurable  for all $n\ge 0$ and $v\in[N]$.
\end{definition}

Based on the independence of $\sigma(A)$ and $\sigma(B)$ whenever $A$,$B$ are disjoint subsets, we can  prove the LTI property of the EN solvency cascade mechanism with fractional recovery.
\begin{proposition}(LTI property of the solvency cascade mechanism) Consider an IRFN conditioned on a skeleton graph $([N],\CE)$ (or equivalently the realized random variables $T,I$) that is a non-random, connected, directed tree. Then the solvency cascade defined by  \eqref{stepn:1}-\eqref{impactedsolvency2}  for any parameter $\lambda\in[0,1]$ is such that for all $n\ge 0$ and $(wv)\in\CE$,
\begin{enumerate}
   \item $\Delta^{(n)}_{v}$ is $\sigma(\CM^-_v\cup\{v\})$-measurable.
 \item $S^{(n)}_{wv}$ is $\sigma(\CM^-_w\cup\CM^+_{w\backslash v}\cup{\{w\}}\cup\{(wv)\})$-measurable.
  \item $S^{(n)}_{v}$ is $\sigma(\CM^-_v)$-measurable.
\end{enumerate}
\end{proposition}

\begin{proof} First we note that for any $(wv)\in\CE$, $\bar X_{w\backslash v}$ is  $\CM^+_{w\backslash v}$-measurable. Next note that $\Delta^{(0)}_{v}$ is $\sigma({\{v\}})$-measurable. Now assume inductively that  $\Delta^{(n)}_{v}$ is $\sigma(\CM^-_v\cup{\{v\}})$-measurable for $n=k$ and all $v$. Then it follows that
\begin{enumerate}
\item For any $(wv)\in\CE$, 
$S^{(k)}_{wv}=G_\lambda(\bar X_{w\backslash v},\Delta^{(k)}_{w},\bar\Omega_{wv})$ which is $\sigma(\CM^+_{w\backslash v}\cup\CM^-_w\cup{\{w\}}\cup\{(wv)\})$-measurable, and hence $\sigma(\CM^-_v)$-measurable by \eqref{disjoint}. 
   \item $S^{(k)}_{v}=\sum_{\{w:(wv)\in\CE\}}\ S^{(k)}_{wv}$ is $\sigma(\CM^-_v)$-measurable.
  \item $\Delta^{(k+1)}_{v}=\Delta^{(0)}_{v}-S^{(k)}_{v}$ is $\sigma(\CM^-_v\cup\{v\})$-measurable.
\end{enumerate}
This verifies the inductive step for $n=k+1$, and hence the proposition is proven for all $n\ge 0$. 
\end{proof}

This proposition unravels the independence relationships across the entire family of balance sheet and exposure random variables that arise as the solvency cascade mapping is iterated. When the LTI property of the solvency cascade mechanism is combined with the fact that the infinite skeleton of an IRFN has components that are all random trees, it is not at all surprising that the large-$N$ asymptotics of the cascade mapping is consistent with Conjecture \ref{Conj}. Moreover, since the LTI property extends for any finite number of cascade steps, we have confidence to the extend the conjecture to all higher orders in the cascade.

\subsection{Higher Order Cascade Steps}

The proposed solvency cascade dynamics is given by iterates $n=0,1,2,\dots$  of the mapping from $\Delta^{(0)}$ to $\Delta^{(1)}$ defined above, assuming the conjectured $N=\infty$ asymptotic approximation.  This dynamics will take the probability distribution data for the collection $\{\Delta^{(n)}_v\}$ to probability distribution data for the collection $\{\Delta^{(n+1)}_v\}$. Given the distributional data for the collection $\{T_v, I_{vw},\bar \Omega_{vw},\Delta^{(n)}_v\}$, step $n$ of the full cascade is therefore generated by the following algorithm.

\bigskip\noindent
{\bf Cascade Mapping:\ } \begin{enumerate}
  \item To compute the univariate CF $\hat f_S^{(n)}(k\mid T)$  of the total solvency shock $S^{(n)}_{1}$, use \eqref{shock2} with $\hat f_\Delta^{(0)}$ replaced by $\hat f_\Delta^{(n)}$: 
   \be \BBE[e^{ikS^{(n)}_{1}}\mid T] =\exp \left(\sum_{T'}\BBP(T')\kappa(T',T) \int^\infty_{-\infty} \ \hat f_\Delta^{(n)}(k'\mid T')\ R(k, k'\mid T,T')\ \dee k'\right)\label{shock2n}
 \ee

   \item To compute the univariate distribution of the impacted default buffer  $\Delta^{(n+1)}_1=\Delta^{(0)}_1-S^{(n)}_1$ use the formula \eqref{newdelta}:
  \be\label{newdelta2} \hat f_\Delta^{(n+1)}(k\mid T) =\hat f_\Delta^{(0)}(k\mid T) \hat f_S^{(n)}(-k\mid T)\ .
\ee
\end{enumerate}

\subsection{Cascade Steps: Algorithmic Complexity}

A numerical implemention of the cascade mapping just described will require suitable truncation and discretization to approximate the integrals  in \eqref{shock2n} by finite sums.  In other words,  we need to find a suitable truncation parameter $L$ and discretization parameter  $\delta $ such that computing the function ${\bf R}(k,k'\mid T,T')=\BBP(T')\kappa(T',T) R(k,k'\mid T,T')$ for $k,k'$ on the grid $\Gamma=\delta\{-L+1/2,-L+3/2,\dots, L-3/2, L-1/2\}^2\subset \BBR^2$ provides sufficient accuracy. 
Note that ${\bf R}$, a square matrix with $2L\times M$ rows and columns, is only computed once for the entire cascade.

Given ${\bf R}$, the algorithm for each cascade step maps the  $2L\times M$ dimensional vector  $\hat f_\Delta^{(n)}$ to the exponential of a matrix product $\hat f_S^{(n)}=\exp[ {\bf R} * \hat f_{\Delta}^{(n)}]$, followed by a Hadamard (element-wise) product
$\hat f_{\Delta}^{(n+1)}=\diag(\hat f_{\Delta}^{(0)})*\hat f_S^{(n)} $. 

Thus the solvency cascade mapping admits a very compact specification in terms of the sequence of conditional characteristic functions, taken as vectors  $\hat f_{\Delta}^{(n)}:={\bf f}^{(n)}\in\BBC^{2L\times M}$, namely:
\be\label{compactCM} {\bf f}^{(n+1)}=\CC({\bf f}^{(n)}):=\diag({\bf f}^{(0)})*\exp[ {\bf R} *{\bf f}^{(n)}]\ .\ee
The nonlinear mapping $\CC:\BBC^{2L\times M}\to\BBC^{2L\times M}$ is parametrized by the solvency cascade kernel $ {\bf R} $ and the default buffer distribution ${\bf f}^{(0)}$, which, we can also note, must satisfy complex conjugation identities $\overline{ {\bf R} (k,k')}= {\bf R} (-k,-k')$ and $\overline{f^{(0)}(k)}=f^{(0)}(-k)$. A single cascade step is thus of order $O(L^2\times M^2)$ flops plus $2L\times M$ ordinary exponentiations.  
 In general, a {\em cascade equilibrium} is a fixed point ${\bf f}^*$ of the mapping, 
\[ {\bf f}^*=  \diag({\bf f}^{(0)})*\exp[ {\bf R} *{\bf f}^{*}]\ .\]

\section{Implementing IRFNs} Consider a generic banking network for some country that consists of $\hat N=\sum_{T\in[M]} \hat N_T$ banks classified into $M$ types labelled by $T\in[M]$,  where $\hat N_T$ denotes the number of banks of type $T$. Suppose the interconnectivity, exposures and balance sheets of the network have been observed monthly for the past $N_m=12$ months. Bank type can be assumed not to change, but the connectivity and balance sheets will fluctuate over the period. The aim here is to construct a sequence of IRFNs of size $N$ increasing to infinity, that is statistically consistent with the real world pre-crisis financial network when $N=\hat N$. Then the statistical model for $N=\infty$ can be subjected to crisis triggers with any type of initial shock $\delta {\rm B}$, and the resultant solvency cascade analytics developed in Section 3 will yield measures of the resilience of the real world network. 

For any of the monthly observations of the network, directed edges are drawn between any ordered pair $(v,w)$ of banks if the exposure of bank $w$ to bank $v$ exceeds a specified threshold (a ``significant exposure''). Let $\hat E=\sum_{T,T'} \hat E_{T,T'}$ be the total number of significant exposures in the network identified in the $N_m=12$ month historical database, decomposed into a sum over the bank types involved. For each  $T\to T'$ edge  $e\in[\hat E_{T,T'}]$ we observe the value $\Omega_e$; For each $v\in[N_m\times \hat N_T]$ we also observe samples $ {\rm B}_v$ of the type $T$ balance sheets. Our large $N$ IRFN will be calibrated to this data.

\subsection{Calibrating the Large N Model} \label{largeNnetwork}
The data described above leads to a natural calibration of the pre-trigger IRFN model for any value of $N$ (including $N=\infty$) at any time in the near future.
A bank $v$ randomly selected from the empirical distribution will have type $T$ with probability
\[ \widehat\BBP(T)=\frac{\hat N_T}{\hat N}\ .\]
Conditioned on $T_v=T$, its balance sheet $ {\rm B}_v=[\bar{\rm A}_v,\bar\Xi_v,\bar\Delta_v]$ will be drawn from the distribution whose {\em empirical characteristic function} is
\be\label{ECF2} \hat f_ {\rm B}({\bf u}\mid T)=\frac1{N_m\times \hat N_{T}}\sum_{v=1}^{N_m\times \hat N_{T}}e^{i{\bf u}\cdot {\bf B}_v}\ee
as a function of $\bf u\in \BBR^3_+$. 

A randomly selected pair of banks $e=(v,w), v\ne w$ with types $T,T'$  respectively will have a significant directed exposure, and hence a directed edge, with probability
\[  \widehat\kappa(T,T')=\frac{\hat E_{T,T'}}{N_m \hat N_T(\hat N_{T'}-\delta_{TT'})}\ .\]
where the matrix $\widehat\kappa$ is called the {\em empirical connection kernel}. Finally, for each ordered pair $T,T'$ we have $\hat E_{T,T'}$ observed significant exposures $\Omega_e$ from a $T$ bank to a $T'$ bank, leading to the empirical characteristic function
\be\label{ECF1} \hat f_\Omega(u\mid T,T')=\frac1{\hat E_{T,T'}}\sum_{e=1}^{\hat E_{T,T'}}e^{iu\Omega_e}\ .\ee

Solvency cascade computations involve integrals over the $u$-variables, which must be approximated by finite sums obtained by truncation and discretization. This will lead essentially to the Fast Fourier Transform, which amounts to choosing a suitably small discretization parameter $\delta$ and  large truncation value $\delta L$ and computing each occurrence of \eqref{ECF1} for the finite lattice $u\in \delta\{-L+1/2, -L+3/2,\dots, L-3/2, L-1/2\}$.

 The increasing sequence of random IRFN models based on these empirical probability distributions is intended to capture essential aspects of systemic risk in our specific finite real world network. For this to be true, a necessary condition to be verified will be that the  $N=\infty$ solvency cascade analytics should also provide a reasonably accurate approximation to simulation results for finite $N$.

\subsection{Parametrization Issues} There are several issues that need to be addressed by extensive experimentation when implementing such a scheme. \begin{enumerate}
  \item Network sparsity: What is the best threshold for defining ``significant exposures''?  There is a tradeoff between increasing the connectivity (reducing sparseness) and the cost of ignoring small exposures:  It has been argued that only ``large exposures'' are important in SR.   Computational burden is not sensitive to the exposure threshold.
  \item How many types of banks is ideal? Again, there is a tradeoff.  Taking $M$ sufficiently large is important because this is the parameter that determines how realistically the network correlation can be modelled.  However, note that the computational burden increases and the power of the statistical estimation decreases with the number of types. 
  \item How large must $N$ be chosen so that the asymptotic analysis is a good approximation? Likely, the accuracy of the large $N$ approximation will deteriorate as the number of types increases. How sensitive is the accuracy of the LTI approximation (which relies to some extent on the sparsity of the network) to the choice of exposure threshold? 
  \item Where can one obtain the data required to calibrate IRFN models? Exposure data with identified counterparties is never publicly available, and currently is often not available even to regulators. So finding real world network data is a serious impediment to implementing any kind of financial network model. 
\end{enumerate}

\subsection{Numerical Experiments} We have described how to implement a generic IRFN solvency cascade model from the point of view of someone with access to complete counterparty-counterparty exposure data. Since such detailed data rarely exists, and is never publicly available, a practical way to gain understanding of the IRFN is to follow the above implementation method for simulated network data. For example, \cite{HurdGleeMeln17}[Section 3.2] investigates zero-recovery solvency cascades in a stylized configuration graph random network with three bank types, that mimics certain characteristics of the US financial network. One can simulate such a model with $N=4000$ banks (roughly the current number of banks in the US), over $N_m=12$ months, and follow the calibration method to match the resulting ``synthetic'' network data to the IRFN framework. Since the model of \cite{HurdGleeMeln17}[Section 3.2] is much simpler than, but not a special case of, the IFRN framework, it is of interest to investigate how features observed in the simple model evolve and change when the new structural elements of the IRFN framework, such as fractional recovery and random exposures, are included. Such simulation-based experiments are easily accessible, and certainly merit future investigation as a testing framework for researchers in systemic risk modelling. 

\section{Conclusion}

This paper concerns itself only with general definitions, characteristics and properties of the IRFN cascade framework. Although the framework is designed to address any of a wide range of systemic risk effects for a wide range of real world financial networks, no attempt is made here to demonstrate its usefulness in actual specific contexts. It goes without saying that extensive and detailed numerical explorations of such implementations are needed to gain evidence that the IRFN method can be a useful and informative guide to understanding systemic risk. Fundamental questions of an implementation nature such as the accuracy of the large $N$ approximations, the development of efficient computation schemes, large scale simulation experiments, and calibration to real network data, are very important but would amount to an enormous expansion of the scope of this paper, and by necessity are postponed to future works. 

As it stands, this paper provides a number of innovative mathematical ideas. The first is that the IRFN framework provides a flexible mathematical representation applicable to real world networks viewed at a suitably coarse grained scale. For example, it can provide a representation of the global financial network that can be useful in  understanding SR spillovers between countries. Not only is the IRFN framework versatile, it possesses an underlying mathematical structure called  the ``locally tree-like independence property'' that means large $N$ asymptotic formulas for the network can be related to an associated Galton-Watson branching process. Some of the mathematical details of this type of ``percolation theory'' remain conjectural, and open to future research. 

A second innovation is the analysis of the EN 2001 default modelling cascade mechanism, and its generalizations, within the IRFN framework. It is shown that these mechanisms possess a related kind of LTI property, that in essence unravels the dependence structure of the sequence of balance sheet random variables arising from the cascade mapping, under the condition that the skeleton connection graph is a random tree.  This property motivates the large $N$ cascade mapping formulas and fixed point equilibrium condition derived in the paper that dramatically extend rigorous large $N$ results of \cite{AminContMinc16} and \cite{DeteMeyePana17}, but whose proof remains another open problem for research. As noted in those works, the biggest conceptual advantage of closed analytical cascade mapping formulas such as these is to provide measures of resilience of the network that depend only on a reduced set of relevant model parameters.

A third important contribution is the extension of random financial networks to involve quite general classes of distributions, such as inhomogeneous random exposures and balance sheets. Moreover, the solvency cascade mapping analysis extends naturally in this wider setting, leading to a remarkably compact formula \eqref{compactCM}.

 Perhaps the key obstacle in systemic risk research is the strategic value and importance of counterparty data that makes it extremely confidential, to the extent that collaboration between countries may seem to be impossible. A fourth contribution of the paper, the calibration method outlined in Section 4, addresses this issue.  This method relies only on aggregated data that is naturally anonymized, which makes shared calibration exercises possible when implementing a carefully designed global  IRFN  model.

Future work on the foundations of the IRFN approach, as opposed to the implementation issues mentioned above, may take several directions. One way to go is to intertwine solvency shocks with funding liquidity shocks as well as indirect channels of contagion. With two or more of the channels of systemic risk, the  LTI  property of the cascade mechanism seems to fail, complicating the large $N$ limit analysis. Another type of extension is to add node types for financial institutions such as hedge funds, firms, central clearinghouses, central banks etc. This can be implemented within the IRFN framework, introducing another dimension of complexity. A third type of extension will be to make the exposures have the meaning of cash, collateral and other types of contract.

 In a nutshell, this paper provides a flexible and convenient framework with many potential applications to systemic risk. However, proving the value of the IRFN approach will depend most heavily on the results that compare  cascade simulations to analytic cascade formulas, for network models calibrated to reflect the properties of real world financial systems.

\bibliographystyle{abbrvnat}


\appendix
\section{Proofs}

  \begin{proof}[Proof of Lemma 1] Under the assumptions, one can show directly that $f(x,y):=\log(1+yg(x,y))]- yg(x,0)$ satisfies $\lim_{y\to 0}f(x,y)=\lim_{y\to 0}\partial_yf(x,y)=0$ and hence by Taylor's remainder theorem
  \[ f(x,y)=\int^y_0(y-v)  \partial^2_yf(x,v)\dee v \]
  One can also show that $\partial^2_yf(x,v)$ is in $L_2(I)$ for each value $v\in[0,\bar y]$ provided $\bar y>0$ is small enough. Then, by Fubini's Theorem, for $y\in[0,\bar y]$
  \[ \||\log(1+yg(x,y))]- g(x,0)\||^2\le (\int^y_0(y-v)\dee v)^2\max_{v\in[0,\bar y]}\||\partial^2_yf(x,v)\||^2\le M y^4\]
  for some constant $M$, from which the result follows.
\end{proof}

\end{document}